\documentclass{acm}
\usepackage{url}
\usepackage{algorithm, algorithmic}
\usepackage{listings}
\usepackage{pgfplots}
\lstMakeShortInline[language=java]{|}

\newcommand{\ol}[1]{\texttt{\small #1}}

\newcommand{\inner}[2]{\left\langle #1,#2 \right\rangle}

\newcommand{\xy}[0]{(\vec{x},y)}
\newcommand{\w}[0]{\vec{w}}
\newcommand{\argmin}[1]{\underset{#1}{\operatorname{argmin}}}

\begin{document}

\title{Iterative MapReduce for Large Scale Machine Learning}

\numberofauthors{3}
\author{
\alignauthor{\text{Joshua Rosen, Neoklis Polyzotis}} \\
\affaddr{University of California, Santa Cruz}
\alignauthor{\text{Vinayak Borkar, Yingyi Bu, Michael J. Carey}} \\
\affaddr{University of California, Irvine}
\and 
\alignauthor{\text{Markus Weimer, Tyson Condie, Raghu Ramakrishnan}} \\
\affaddr{Yahoo! Research} 
}

\maketitle

\begin{abstract}
  Large datasets (``Big Data'') are becoming ubiquitous because the
  potential value in deriving insights from data, across a wide range
  of business and scientific applications, is increasingly recognized.
  The data growth has been accompanied by rapid adoption of large,
  elastic, multi-tenanted computing clusters (``compute clouds''),
  leading to a virtuous cycle: the scalability of cloud computing
  makes it possible to analyze ever larger datasets, and the
  proliferation of Big Data leads to further adoption of cloud
  computing.  In particular, machine learning---one of the
  foundational disciplines for data analysis, summarization and
  inference---on Big Data has become routine at most organizations
  that operate large clouds, usually based on systems such as Hadoop
  that support the MapReduce programming paradigm.  It is now widely
  recognized that while MapReduce is highly scalable, it suffers from
  a critical weakness for machine learning: it does not support
  iteration.  Consequently, one has to program around this limitation,
  leading to fragile, inefficient code.  Further, reliance on the
  programmer is inherently flawed in a multi-tenanted cloud
  environment, since the programmer does not have visibility into the
  state of the system when his or her program executes. Prior work has
  sought to address this problem by either developing specialized
  systems aimed at stylized applications, or by augmenting MapReduce
  with ad hoc support for saving state across iterations (driven by an
  external loop).  In this paper, we advocate support for looping as a
  first-class construct, and propose an extension of the MapReduce
  programming paradigm called {\em Iterative MapReduce}.  We then
  develop an optimizer for a class of Iterative MapReduce programs
  that cover most machine learning techniques, provide theoretical
  justifications for the key optimization steps, and empirically
  demonstrate that system-optimized programs for significant machine
  learning tasks are competitive with state-of-the-art specialized
  solutions.
\end{abstract}
\terms{Systems, Machine Learning}
\section{Introduction}\label{sec:intro}
The volume of data is skyrocketing as organizations recognize the potential
value of data-driven approaches to optimizing every aspect of their operation,
and scientific disciplines ranging from astronomy to zoology become
increasingly data-centric in everything from hypothesis formulation to theory
validation.  Large scale analytics are a key to deriving insight from this
deluge of data, and {\em machine learning} (ML) is now established as a
foundational discipline that is ever more valuable as datasets grow
larger~\cite{allreduce}.  For example, by analyzing billions of transactions,
credit-card companies are able to quickly identify stolen credit card;
insurance companies derive can flag claims for possible fraud.  Supermarkets
derive promotions based on consumer purchases.



The sheer size of today's data sets far exceeds the capacity of a
single machine.  Big Data analytics platforms based on the MapReduce
paradigm, such as Hadoop, have enabled statistical queries over large
data, and many ML algorithms can be cast in terms of these queries
\cite{kearns93t,Chu:2006fk}.  However, MapReduce fails to recognize
the iterative nature of most ML algorithms, and due to this
unfortunate omission, while ML computations can be expressed using
MapReduce, execution overheads are significantly higher than in
Message Passing Interface or algorithm-specific implementations
(e.g.~\cite{Ye:2009zr, Weimer:2010fk}).


Failing to recognize iteration as a first class programming
abstraction is a step backwards, as it forces the programmer to make
systems-level decisions.  For example, in Spark~\cite{Zaharia:2010uq}
the programmer has to decide what data to cache in distributed main
memory.  This approach is ill-suited for large, multi-tenant clusters
such as public clouds where important performance-related parameters
change constantly and in a way that is hard for a programmer to track.
In addressing this challenge, we draw our inspiration from database
systems, where the level of abstraction introduced by the relational
model freed users from low-level systems considerations, and opened
the door to DBMS-driven optimization.

In this paper, we present extend the MapReduce paradigm with support
for iteration, and present a principled framework for optimizing the
runtime of systems such as Hadoop to efficiently support Iterative
MapReduce programs.  To this end, we make the following contributions:

\begin{enumerate}
\item {\bf Iterative MapReduce: } We formalize the Iterative MapReduce
  programming model, and describe how many recent proposals to support
  ML over Big Data can be expressed readily in this
  model. (Section~\ref{sec:iterative-map-reduce})
\item {\bf Runtime: } We present a new runtime for Iterative
  MapReduce. (Section~\ref{sec:physical-plan})
\item {\bf Optimizer:} We develop an optimizer that picks a good
  runtime plan when given data, program and cluster parameters.  In
  particular, we consider two key choices: the partitioning strategy
  for the training data, and the structure of the aggregation that is
  applied to the intermediate statistics produced by the
  computation. We argue that these are the only tunable knobs since
  the computation itself (the logic of the Map and Reduce steps) is
  opaque, and present a theoretical foundation for our
  optimizer. (Section~\ref{sec:optimizer})
\item {\bf Empirical study: } We empirically validate both our
  optimizer and our runtime, the latter by demonstrating that it can
  outperform a state of the art system,
  VW~\cite{allreduce}. (Section~\ref{sec:experiments})
\end{enumerate}


\section{Iterative MapReduce}\label{sec:iterative-map-reduce}\subsection{Background: MapReduce}

MapReduce is a functional programming model that splits the
traditional group-by-aggregation computation into two steps: \ol{map}
and \ol{reduce} \cite{Dean:2004uq}.  The (user-specified, opaque)
\ol{map} step is responsible for transforming the input into {\em
  key-value} record pairs.  The {\em key} identifies the group to
which the {\em value} belongs; all values associated with the same key
are grouped together.  The (also user-specified and opaque)
\ol{reduce} step is then used to process each group and produce the
final result.  The computation associated with the \ol{reduce} step is
commonly an aggregate function (e.g., sum, max, mean, etc.), which
produces a scalar value for each group.

The MapReduce programming model has been used to implement many
higher-level programming abstractions.  Pig Latin~\cite{Olston:2008kx}
and Hive~\cite{hive} both provide a SQL layer with some notable
extensions (e.g., correlated sub-queries) on top of the Hadoop
MapReduce runtime.  Such higher-level abstractions allow programmers
to express their computations in a form that is closer to an intended
target domain (e.g., data analytics).  In our work, we have built a
higher-level abstraction for machine learning using MapReduce called
\emph{ScalOps}~\cite{Weimer:2011fk}, which is a Scala domain-specific
language (DSL) that uses Pig Latin like syntax.

\subsection{Iterative MapReduce}
Many machine learning algorithms can be expressed as iterative
procedures refining the model, given training data. More to the point,
the body of these iterations can be expressed solely in terms of
\emph{statistical queries}~\cite{kearns93t} over the training data
such as min, max, mean and sums; these queries can be naturally
computed in MapReduce.  This insight was used by Chu~et
al.~\cite{Chu:2006fk} to express effective parallel versions of
several machine learning algorithms (e.g., backpropagation in neural
networks, EM, logistic regression, linear SVMs, PCA) relying only on
sums over functions applied to the data.

Inspired by these earlier results and building on our own work towards
a more general programming interface for cloud-based Big Data analysis
\cite{Weimer:2010fk,Weimer:2011fk}, we introduce an extension of the
MapReduce programming paradigm, called \emph{Iterative MapReduce},
that supports iteration as a fundamental construct.  Iterative
MapReduce is defined in terms of a collection of operators that can be
composed to create dataflow programs.  Each Operator accepts an input
and produces an output.  Chaining operators therefore is the main
composition method in Iterative MapReduce.  The computation itself is
expressed in these three key operators:

\begin{description}
\item[\ol{MapReduce}:] This operator has two inputs: the data set and
  side information that it makes available to the user defined
  \ol{map} and \ol{reduce} functions it hosts. The \ol{map} function
  is applied to all records in the immutable input data and the
  \ol{reduce} function is applied to aggregate the outputs of that
  process. We define \ol{reduce} in the sense typically found in
  functional programming languages: It is a associative and cumulative
  function that accepts two inputs and reduces them to a single
  output.  Section~\ref{sec:physical-plan} looks at how we can
  parallelize this step over a cluster of machines.
\item [\ol{Sequential}:] This operator accepts a single input, and
  produces a single output using the user defined function it hosts.
  Separating such code from the \ol{MapReduce} operator allows us to
  ensure an associative and commutative \ol{reduce} function.
\item[\ol{Loop}:] This is a fundamental extension to the basic
  MapReduce paradigm.  As in most programming languages, our \ol{Loop}
  operator accepts three inputs: a body, a condition and an
  initializer.  The \emph{body} contains a chain of \ol{MapReduce} and
  \ol{Sequential} operators. The output of one forms the input of the
  next operator in this chain.  We require that the output of the last
  operator is valid input to both the loop condition (see below) and
  the first operator of the chain.  The \emph{condition} accepts the
  loop body's output as input and returns a boolean value indicating
  whether the loop should terminate, while the \emph{initializer} is
  used to provide an initial input for the loop body.
\end{description}

Many programs can be expressed using these three operators.
Trivially, they facilitate the construction of loops over sequential
code.  More importantly, they allow us to write iterative ML
algorithms without recourse to external mechanisms (in particular,
without using a top-level driver that invokes MapReduce within a loop,
but is not visible to the MapReduce system).  To express most
iterative ML algorithms, the loop body would consist of a single
\ol{MapReduce} operator that computes the relevant statistics, using
the current model state as an input.  This would be followed by a
\ol{Sequential} operator that updates the model.

While we discuss this special case extensively due to its importance
in the machine learning domain, we note that the Iterative MapReduce
programming model is in fact more general, and supports loops over
multiple \ol{MapReduce} operators as well as loops over any sequence
of \ol{MapReduce} and \ol{Sequential} operators.  This, for example,
allows facilitates the native expression of optimization algorithms
that probe multiple possible gradient step sizes.

\begin{figure}
  \begin{center}
    \includegraphics[scale=0.45]{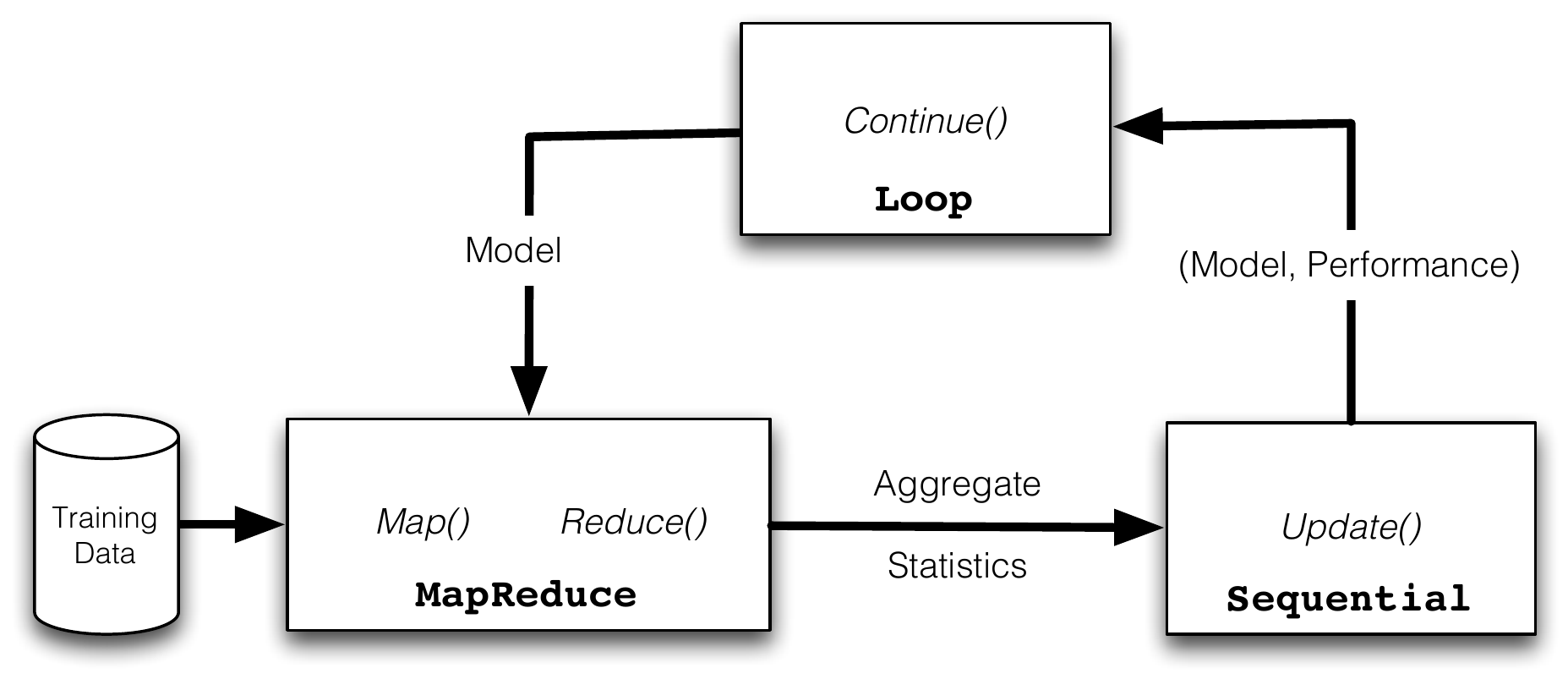}
    \caption{\label{fig:imr-dataflow}Iterative MapReduce Dataflow for ML.}
  \end{center}
\end{figure}

The situation for the important special case is depicted in
Figure~\ref{fig:imr-dataflow} from a data-flow perspective.  The
arrows indicate control flow, which carry data from one step to the
next.  The {\bf Loop} operator drives each iteration, until some
stopping condition is met.  It is also responsible for producing the
initial model.  In a single iteration, the {\bf MapReduce} operator
accepts the current model and uses it to process the training data,
and produce an aggregate statistic.  The {\bf Sequential} step uses
the aggregate statistic to update the model, before returning control
to the {\bf Loop} for a (possible) subsequent iteration.

While we are not the first to recognize the need for supporting iteration in
MapReduce, we are the first to explore the consequences of adding iteration as
a fundamental construct in the MapReduce system, and in particular to
demonstrate the opportunities for system-driven program optimization.  Prior
work is focused on assembling ML algorithms within a specialized
runtime~\cite{allreduce,Malewicz:2009ve,Zaharia:2010uq} targeted at specific
applications, or invoking a general purpose MapReduce engine~\cite{Chu:2006fk}.
In contrast, we have developed the Iterative MapReduce programming model with
ML-style programs in mind (the \ol{Loop} operator is especially noteworthy),
and (in Section~\ref{sec:optimizer}) develop an optimizer that can translate a
broad class of programs in this model (covering most ML programs) to an
efficient runtime execution plan for an arbitrary cluster environment.  The
systems-driven optimization enabled by our approach is especially valuable in
multi-tenanted and elastic cloud systems, whose rapidly changing resource
availability makes it difficult if not impossible for programmers to manually
configure their programs effectively.


\section{Related Work}\label{sec:related-work}
In translating programs from our programming model to efficient runtime plans,
we seek to exploit optimizations discovered in prior work, which we review in
this section.

{\bf Hadoop~\cite{hadoop}} is the dominant Open Source Software implementation
that supports the MapReduce programming model~\cite{Dean:2004uq}.  A Hadoop job
executes a single MapReduce iteration.  The input and output of the job is
stored in a distributed filesystem (HDFS).  A job consists of a \ol{map} and
\ol{reduce} step, which are parallelized over many tasks.  Hadoop tries to
schedule map tasks on machines that host the input data, so the number of map
tasks is data dependent.  The number of reduce tasks is a job parameter, set by
the programmer.  The intermediate data produced by the map tasks and consumed
by the reduce tasks is managed by the Hadoop runtime, which uses a sort-based
implementation to perform the group-by operation.  The Hadoop API also exposes
a ``combiner'' function that supports pre-aggregation of this intermediate
data.  Hadoop does not have support for a \ol{loop} step.  Instead, an external
driver must implement such a step by repeatedly submitting jobs to the Hadoop
runtime.  Each job executes in isolation and any information produced by the
previous job is fed to the new job through back channels (i.e., the HDFS file
system).  Lastly, the training data must be re-read from its source (i.e.,
HDFS), forgoing the benefits of caching.

{\bf HaLoop~\cite{Yingyi-Bu:2010fk}} exposes an application programming
interface that supports iterations in Hadoop MapReduce.  The extension adds a
loop control module to the Hadoop master node that repeatedly spawns new jobs
based on a loop body, until come stopping condition is met.  HaLoop also adds
cache aware scheduler to Hadoop that colocates map tasks with the reduce task
that produces its input.

{\bf MPI Launchers} (e.g., \cite{Ye:2009zr}) address the need for an explicit
\ol{loop} step, and by doing so, avoid the scheduling overheads observed in
Hadoop.  Pregel~\cite{Malewicz:2009ve} and Giraph~\cite{giraph} are two recent
runtimes that support a message passing interface (MPI) programming model.
Both systems expose an API for loading and caching input data.  The \ol{map}
step is automatically fed the output of the prior iteration, usually in the
form of messages.  The \ol{reduce} step is supported by global ``aggregators.''

{\bf Worker-Aggregator~\cite{Weimer:2010fk}} defined by Weimer et al., is a
\emph{distributed main memory} implementation that uses a \emph{flat
aggregation hierarchy} with a single aggregator task with direct network
connections.  The system outperforms Hadoop by an order of magnitude on a
stochastic gradient descent~(SGD) algorithm.  This speedup is in line with
earlier MPI results~\cite{Ye:2009zr}.  The authors point to a rather unorthodox
handling of failures: As the algorithm evaluated (SGD) is inherently stochastic
in its data access, machine failures can simply be \emph{ignored}, as long as
they occur independently of the data stored on those machines.

{\bf Vowpal Wabbit (VW)~\cite{allreduce}} is a scalable machine learning system
that integrates the machine learning algorithm(s) into the runtime.  The system
includes a Hadoop-aware version of the allreduce function found in MPI.  The
system is highly optimized for fast iterations.  A cache aware data format is
used to speed up the \ol{map} step, and a \emph{binary aggregation tree} is a
key optimization used to speed up the \ol{reduce} step.  Task re-scheduling is
avoided between iterations and communication happens via direct network
connections.  The authors observe an order of magnitude speedup when comparing
with stock Hadoop.

{\bf Spark~\cite{Zaharia:2010uq}} is a runtime built on a data abstraction
called resilient distributed datasets (RDDs) that reference immutable data
collections.  Spark also provides a domain-specific language (DSL) that
consists of standard relational algebra {\em transformations} (select, project,
join) and {\em actions} that perform global aggregation.  Spark supports
iterative algorithms that explicitly cache RDDs in-memory.  Indeed, the Spark
runtime is optimized for in-memory computation only.  Spark has published
speed-ups of $30\times$ over stock Hadoop.

\subsection{Discussion}
Many of the approaches described above claim order of magnitude
speedups over stock Hadoop when performing Iterative MapReduce.  These
runtimes share several characteristics in order to accomplish this
goal.  They avoid rescheduling of machines between iterations, cache
partitioned data between iterations, and use more powerful forms of
aggregation between map and reduce steps.  However, these improvements
have not been cast in a form that can be exploited on an arbitrary
cluster environment.  To do so requires us to capture all significant
aspects of the computation including iteration in the programming model;
develop a formalization of the plan
space, including a definition of runtime operations and key parameters
such as {\em partition width} and {\em aggregation tree fan-in}, in order
to reason about alternative equivalent execution plans; and to build an
optimizer that can evaluate the cost of these alternative plans and choose
a good plan.  We have already introduced the Iterative MapReduce programming
model, which captures iteration; next, we will build on this to formalize the space
of equivalent runtime plans for a given program.  After that, we describe the
optimizer in Section \ref{sec:optimizer}.

\section{Physical Plan}\label{sec:physical-plan}This section presents the physical plans that execute our Iterative
MapReduce programming model on a cluster of machines.  For
concreteness, we consider the Iterative MapReduce dataflow shown in
Figure \ref{fig:imr-dataflow} and discuss a {\em plan template} for
it: the space of equivalent plans is realized by instantiating this
template with different plan parameter values.  Our implementation
uses the Hyracks runtime~\cite{Borkar:2011ly}, and a plan consists of
dataflow processing elements, or Hyracks operators, that execute in
the Hyracks runtime.  Hyracks splits each Hyracks operator into
multiple tasks that execute in parallel on a distributed set of
machines.  Similar to Hadoop, each task operates on a single partition
of the input data.  In Section~\ref{sec:imr-physical-plan}, we
describe the structure of the physical plan template and discuss its
tunable parameters.  Section~\ref{sec:plan-space} then explores the
space of choices that can be made when executing this physical plan on
an arbitrary cluster with given resources and input data.

\subsection{Iterative MapReduce Physical Plan}\label{sec:imr-physical-plan}

\begin{figure}
  \begin{center}
   \includegraphics[width=\columnwidth]{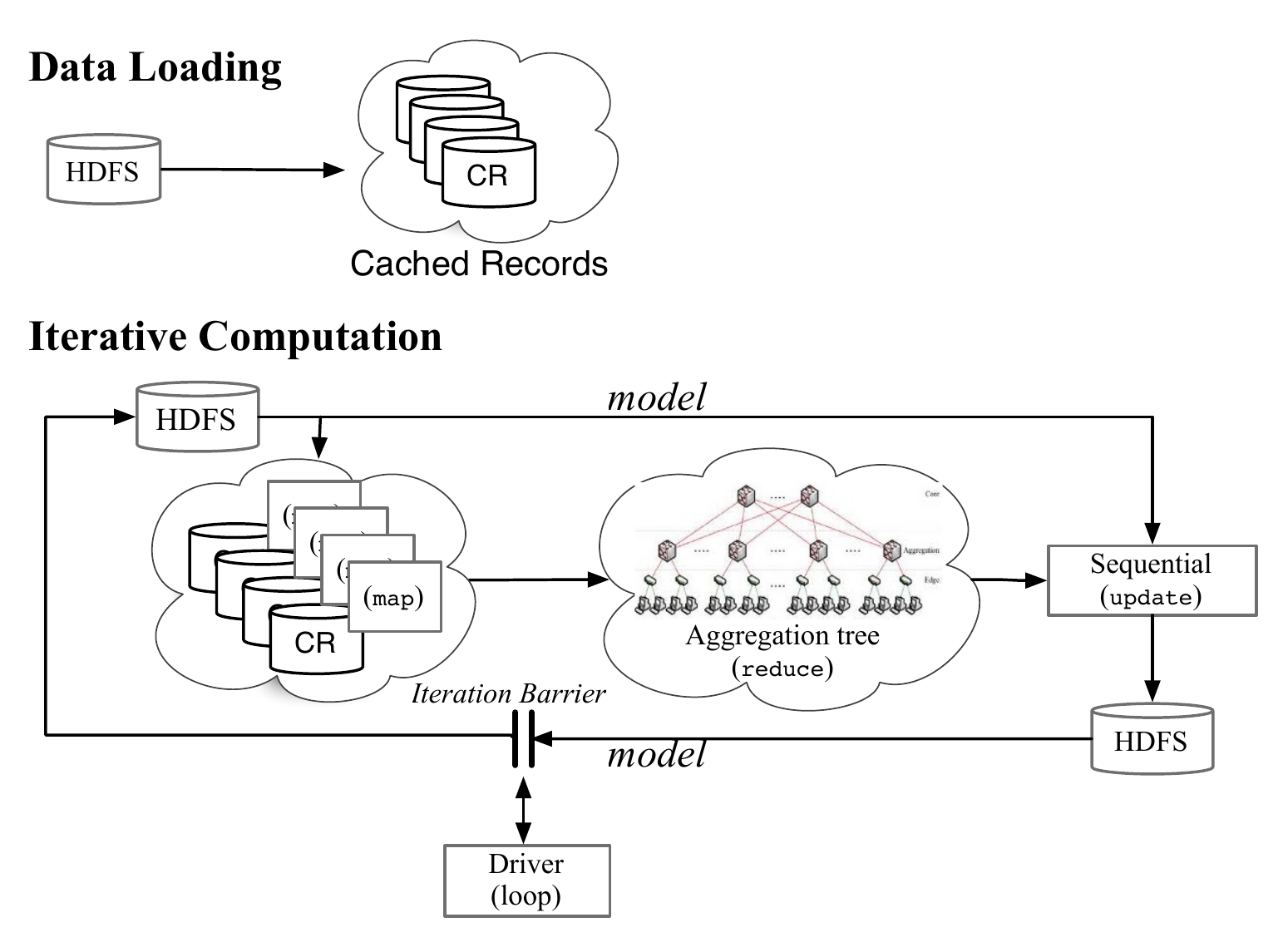}
   \caption{\label{fig:physical-plan}Hyracks physical plan for Iterative MapReduce.}
  \end{center}
\end{figure}

Figure~\ref{fig:physical-plan} depicts the physical plan template for
the Iterative MapReduce dataflow in Figure \ref{fig:imr-dataflow} as
two data-flows.  The top dataflow loads the input data from HDFS,
parses it into an internal representation (e.g., binary formated
features), and partitions it over~$N$ cluster machines.  The bottom
dataflow executes the computation associated with a {\bf Loop}
operator.  The \emph{Driver} of the loop (observe that this is now
controlled by the system, which is now aware of the entire program
including the iteration!) is responsible for seeding the initial
global model and driving each iteration based on the loop condition.
The \ol{map} step is parallelized across some number of nodes in the
cluster determined by the optimizer~\footnote{Reflected in the
  partitioning strategy chosen for the data loading step}.  Each
\ol{map} task sends a data statistic (e.g., a loss and gradient in a
BGD computation) to a random reduce task participating in the
leaf-level of an aggregation tree.  This aggregation tree is balanced
with a parameterized fan-in~$f$ (e.g., $f = 2$ yields a binary tree)
determined by the optimizer.  The final aggregate statistic is passed
to the {\bf Sequential} operator, which updates the global model
stored in HDFS.  The {\em Driver} detects this update, and applies the
loop condition to the new model to determine if another iteration
should be performed.

This description of the plan template highlights two choices to be
determined by the optimizer---the number of nodes allocated for the
map phase of the computation, and the fan-in of the aggregation tree
for the reduce phase.  The structure of the plan template comes from
consideration of the structure of the dataflow in Figure
\ref{fig:imr-dataflow}, and the justification for the focus on these
two optimizer choices will be presented next.


\subsection{The Plan Space}\label{sec:plan-space}

There are several considerations that must be taken into account when
mapping the physical plan in Figure~\ref{fig:physical-plan} to an
actual cluster of machines.  Many of these considerations are
well-established techniques for executing data-parallel operators on a
cluster of machines, and are largely independent of the resources
available and the program/dataset to be optimized.  We begin by
discussing these ``universal'' optimizations for arriving at an
execution plan.  Next, we examine those choices that are dependent on
the cluster configuration (i.e., amount of resources) and computation
parameters (i.e., input data and aggregate value sizes).  These are
the choices an optimizer must make for a given program and input
dataset in the context of a given cluster and current workload.

\subsubsection{Universal Optimizations}
\label{sec:univopt}

\textbf{Data-local scheduling} is generally considered an optimal
choice for executing a dataflow of operators in a cluster environment:
a map task is therefore scheduled on the machine that hosts its input
data.  \textbf{Loop-aware scheduling} ensures that the task state is
preserved across iterations.  Note that this is not the same as
blocking machines, as is done in VW~\cite{allreduce}.  Rather, we want
to avoid costly re-optimization per-iteration, taking advantage of the
similarity between iterations.  \textbf{Caching} of immutable data can
offer significant speed-ups between iterations.  However, careful
consideration is required when the available resources do not allow
for such caching.  For example, it is assumed in~\cite{Zaharia:2010uq}
that sufficient main memory is always available to cache the data to
be saved across iterations, and performance degrades rapidly when this
assumption does not hold.  \textbf{Efficient data serialization} can
offer significant performance improvements.  We use a binary formated
file, which has substantial benefits in terms of space and time over
simple Java objects, to store our cached records.

\subsubsection{Per-Program Optimizer Decisions}

The optimizations discussed in Section \ref{sec:univopt} apply equally
to all jobs and can be considered best practices inspired by the
best-performing systems in the literature.  This leaves us with two
optimization decisions that are dependent on the cluster and
computation parameters; we discuss them below.  In the next section,
we develop a theoretical foundation for an optimizer that can make
these choices effectively.

\textbf{Data partitioning} determines the number of \ol{map} tasks in
an Iterative MapReduce physical plan.  For a given job and a maximum
number~$N_{max}$ of machines available to it, the optimizer needs to
decide which number$N<=N_{max}$ of machines to request for the job.
The decision is not trivial, even ignoring the multi-job nature of
today's clusters: More machines reduce the time in the map phase but
increase the cost of the reduce phase, since more objects need to be
aggregated.  The goal of data partitioning is to find the right
trade-off between \ol{map} and \ol{reduce} costs.

\textbf{Aggregation tree structure} involves finding the optimal
fan-in of a single \ol{reduce} node in a balanced aggregation tree.
Aggregation trees are commonly used to parallelize the \ol{reduce}
function.  For example, Hadoop uses a \ol{combiner} interface to
perform a single level aggregation tree, and Vowpal Wabbit uses a
binary aggregation tree.  In this next section, we develop an
optimizer to decide an optimal tree structure for a given job based on
the fan-in~$f$ of the aggregation nodes in the tree.

\section{Runtime Optimization}\label{sec:optimizer}\newtheorem{theorem}{Theorem}
\newtheorem{corollary}{Corollary}

After factoring out optimizations that are universal in nature, the
optimizer needs to answer two crucial questions for a given job in a
given shared cluster environment: (a) How many machines should we
devote to the task? (b) What fan-in~$f$ should we use for the
aggregation tree phase?  In answering these questions an optimizer can
consider two different objectives: (a) Minimize the response time
(wall-clock time) for the program.  (b) Minimize the cost of the
job. Here, we consider \emph{machine time} as a proxy for cost. While
many other metrics are conceivable in principle, public clouds such as
Amazon EC2 have opted for machine time, which makes it the prime
candidate for minimization.

Below, we present our theoretical findings for these questions for two
cases. First, we show that the optimal fan-in of the aggregation tree
is independent of both the cluster and the job.  We use this result to
design the optimal partitioning for two cases: (a) The per-record
processing time is independent of the number of machines used; this is
the case for systems where either all records are read from disk
(e.g., Hadoop) or all records are held in distributed main memory
(e.g., Spark).  (b) Caching influences the time to access/process a
record, which is at the heart of Iterative MapReduce optimization.

As before, we consider the following simple program expressible in our
programming model: A \ol{Loop} containing a single \ol{MapReduce}
operator followed by a \ol{Sequential} operator.  The time spent in
the \ol{Sequential} operator and the iteration control are small
relative to the time spent on MapReduce, hence the optimizer needs
only to consider the time spent inside of MapReduce operator.

We assume that both our network and computation behave linearly: If we
invoke a UDF twice as often, we assume that it will take twice as
long.  We assume that data transmission to/from a machine behaves
linearly.  When a machine sends or receives data it does so
sequentially.  Both of these assumptions can be violated in real world
clusters under extreme load. However, they represent the behavior
within the optimal load region of the cluster.

These assumptions allow us to use the notation found in
Table~\ref{tab:symbols} to express our model for both the iteration
time and cost.  $M$, $P$ and $D$ can be measured for a given cluster
and job and $R$ is known for a job.

\begin{table}
  \centering
  \begin{tabular}{cl}
    \textbf{Symbol} & \textbf{Meaning} \\
    \hline
    $R$      & total \# records \\
    $N_{max}$ & Max \# CPUs \\
    $M$ & \# records cached per CPU \\
    $P$ & Map time per record \\
    $D$ & Load time per record \\
    $A$ & Aggregation time per object \\
  \end{tabular}
  \caption{Symbols used in the derivations}
  \label{tab:symbols}
\end{table}

Lastly, we assume both the cost and the computational time of the
\ol{MapReduce} operator to be comprised additively of the cost (time)
of the map phase and the cost (time) of the reduce phase. Hence, we
state:
\begin{eqnarray*}
  T(N,f) &=& T_A(N,f) + T_M(N)\\
  C(N,f) &=& C_A(N,f) + C_M(N)
\end{eqnarray*}
As already stated in the equation, we assume the aggregation
time~$T_A$ and cost~$C_A$ to depend on both the fan-in~$f$ and the
number~$N$ of machines used. The time~$T_M$ and cost~$C_M$ to map, on
the other hand, solely depend on the number of machines used.
Intuitively, more machines introduce greater parallelism but at the
same time incur additional aggregation time and cost.

In the remainder of this section, we present theoretically optimal
choices for the fan-in~$f$ and the number of machines~$N$ to be used,
starting with the fan-in.


\subsection{Optimal Aggregation Tree Fan-In}

\begin{theorem}
  The fan-in of the \emph{fastest} aggregation tree is:
  \begin{displaymath}
    \hat{f} = e
  \end{displaymath}
\end{theorem}

\begin{proof}
  The time it takes to aggregate~$N$ inputs in an aggregation tree of
  fan-in $f$ can be phrased as:
  \begin{displaymath}
    T_A(N, f) = A  f  h(N,f) 
  \end{displaymath} 
  where $h(N,f)$ is the number of levels in the tree. Aggregation
  happens in parallel at each level.  Hence, the time per level is the
  time spent in a single aggregation node, $A f$.  The height of a
  tree with $N$ leaf nodes and arity $f$ is
  $h(N,f)=\log_f{n}=\frac{\ln{N}}{\ln{f}}$.  Hence, we arrive at:
  \begin{displaymath}
    \hat{f} = \argmin{f}\left(A\ln(N) \frac{f}{\ln f}\right) =e 
  \end{displaymath}
\end{proof}

\begin{corollary}
  The minimal time process $N$ inputs in a balanced aggregation tree
  is:
  \begin{displaymath}
    \hat{T}_A(N) = A e \ln(N)
  \end{displaymath}
\end{corollary}

\textbf{Intuition: } The independence of the number of inputs is easy
to see: the difference between the optimal aggregation tree for a
small vs. large number of leaf nodes is sheer scaling, a process for
which the arity of the tree does not change. The independence of the
transfer and aggregation time~$A$ is similarly intuitive, as the time
spent per aggregation tree level and the number of levels balance each
other out.

Now we consider cost-optimal aggregation trees. First, we discuss the
static case where the \ol{MapReduce} operator is not part of a
\ol{Loop}.

\begin{theorem}
  The \textbf{cost-optimal} fan-in for the reduce phase of a
  \ol{MapReduce} operator is $N$.
\end{theorem}
\begin{proof}
  Decreasing the fan-in below $N$ introduces additional aggregation
  work and doing so does not decrease the computational cost of the
  reduce operation.
\end{proof}
Consider the case where the \ol{MapReduce} operator is part of a
\ol{Loop}: All machines used need to wait while the aggregation is
running, as it is a blocking operation.
\begin{theorem}
  The \textbf{cost-optimal} fan-in for the reduce phase of a
  \ol{MapReduce} operator inside of a \ol{Loop} is $e$.
\end{theorem}
\begin{proof}
  While the aggregation tree is running, the $N$ map machines are
  idle. The number of inner nodes in the tree is $\frac{N-1}{f-1}$,
  which means that the cost of the idling machines always trumps the
  cost of the aggregation machines. Hence, the fastest aggregation
  tree is also cost-optimal.
\end{proof}
The above establishes that neither the time nor the cost of an
iteration depend on the fan-in~$f$, as we can replace it with its
respective optimal choice of $e$ or $N$. Hence, we can refine our cost
and time model to be solely dependent on the number of machines
used~$N$:
\begin{eqnarray*}
  T(N) &=& T_A(N) + T_M(N)\\
  C(N) &=& C_A(N) + C_M(N)
\end{eqnarray*}

\subsection{Optimal Partitioning}
We use this model to study the optimal choice for $N$.  In Iterative
MapReduce, this choice is complicated by caching effects when compared
to MapReduce: Our physical plan makes sure that as much of the
training data stays available in main memory of the machines as
possible, which speeds up all but the first iteration. However, it is
neither guaranteed that all data can fit into the aggregate main
memory of a cluster, nor that that solution is optimal in terms of
response time or cost. Thus, an optimizer must consider these two
distinct possibilities: (a) the optimal $N$ is the one where all data
fits into the collective main memory, that is $R\leq MN$. (b) Some of
the data is spilled to disk, $R > MN$.

\subsubsection{Response Time Minimization}
\begin{theorem}
  Let $R\leq MN$. The \textbf{time-optimal} number of machines for the
  map phase of a MapReduce operator is:
  \begin{displaymath}
    \hat{N} = \frac{R P}{A e}
  \end{displaymath}
\end{theorem}
\begin{proof}
  The map phase is perfectly parallel. Hence, the total processing
  time is given by:
  \begin{displaymath}
    T(n) = \frac{R}{N}P +  A e \ln(N)
  \end{displaymath}  
  This is minimized for
  \begin{displaymath}
    \hat{N} = \argmin{N}\left(\frac{R}{N}P + A e \ln(N)\right)=\argmin{N}\left(\frac{W}{N} + \ln(N)\right)
  \end{displaymath}
  where $W=\frac{R P}{A e}$.  This is minimized when its first
  derivative $\frac{N-W}{N^2}=0$, which the case for
  $\hat{N}=W=\frac{R P}{A e}$.
\end{proof}
\begin{theorem}
  For $R > MN$, the \textbf{time-optimal} number of machines to be
  used for a MapReduce operator is:
  \begin{displaymath}
    \hat{N}= \frac{RD+RP}{Ae}
  \end{displaymath}
\end{theorem}

\begin{proof}
  Processing all $R$~input records takes $RP$ time. $R-MN$ records
  need to be fetched from disk, which incurs an additional delay of
  $(R-MN)D$. The total time for one iteration is thus given by:
  \begin{displaymath}
    T_2(N) = e A \ln(N) + \frac{RD+RP}{N} -MD
  \end{displaymath}
  The constant $MD$ does not affect the minimizer $\hat{N_2}$ which is
  given, similarly to the analysis above for the case with no
  spilling, for $\hat{N_2}= \frac{RD+RP}{Ae}$
\end{proof}

Our optimizer evaluates both $T_1{\hat{N_1}}$ and $T_2\hat{N_2}$ and
chooses the lower one for the runtime plan.

The number of available machines in a cloud is essentially
unbounded. At the very least, we can assume that the number of
machines available exceeds the number of machines needed to cache all
records of a given job.  Hence, the legitimate question arises whether
such a in-memory solution can ever be slower than a solution using
secondary memory. Below, we study this question.

\begin{theorem}
  Incurring disk I/O is time-efficient, if
  \begin{displaymath}
    \frac{D}{P} \in (0, e^{1 - \frac{MP}{Ae}} - 1)
  \end{displaymath}
\end{theorem}

\begin{proof}
  The spilling configuration is better than the in-memory
  configuration when the best time for the spilling case is better
  than the best time for the in-memory case. i.e.
  \begin{eqnarray*}
    T_2(\hat{N_2}) &<& T_1(\hat{N_1})\\
    A e \ln{\frac{R(D + P)}{A}} - MD &<& Ae \ln{\frac{RP}{A}}\\
    Ae \ln{\frac{D + P}{P}} &<& MD
  \end{eqnarray*}
  Also for spilling to be necessary we know $R > M\hat{N_2}$:
  \begin{eqnarray*}
    R &>& M\frac{R(D + P)}{Ae}\\
    MD&<& Ae - MP\\
    Ae \ln{\frac{D + P}{P}} &<& Ae - MP
  \end{eqnarray*}  
  Hence, we arrive at:
  \begin{displaymath}
    \ln{\frac{D + P}{P}} < 1 - \frac{MP}{Ae}
    \label{eq:io_fast_cond}
  \end{displaymath}
  The above inequality has solutions only when $\frac{MP}{Ae} \in (0,
  1)$. Intuitively, this means that processing all in-memory records
  in, one machine must be cheaper than the time spent by an aggregator
  in receiving all its input aggregate objects. Hence,
  Equation~\ref{eq:io_fast_cond} indicates that when
  \begin{displaymath}
    \frac{D}{P} \in (0, e^{1 - \frac{MP}{Ae}} - 1)
  \end{displaymath}
  allowing some I/O is better than using more machines to facilitate a
  completely in-memory map task.
\end{proof}

\subsubsection{Cost Minimization}
As before, we define cost as the time the iteration takes times the
number of machines used. Again, we need to consider the two cases for
whether or not all data can be held in distributed main memory
separately.

\begin{theorem}
  With $R\leq MN$ the \textbf{cost-minimizing} number of machines to
  use in a MapReduce operator is:
  \begin{displaymath}
    \hat{N_1}=\frac{R}{M} 
  \end{displaymath}
\end{theorem}

\begin{proof}
  Following the discussion above, the iteration cost is given by:
  \begin{displaymath}
    C_1(N) =  e A N \ln(N) + RP
  \end{displaymath}
  Where $e A N \ln(N)$ is the cost of the optimal aggregation tree in
  the Iterative MapReduce setting.  This is minimized for
  $N=0$. However, we know that $R\leq MN$. Hence
  $\hat{N_1}=\frac{R}{M} $ is the minimizer within the domain of $N$.
\end{proof}

\begin{theorem}
  For $R< MN$ the \textbf{cost-minimizing} number of machines to use
  in a MapReduce operator is
  \begin{displaymath}
    \hat{N_2} = e^{\frac{MD}{Ae}}
  \end{displaymath}
\end{theorem}
\begin{proof}
  The cost is given by the cost of the fastest aggregation tree plus
  the cost of the map phase:
  \begin{displaymath}
    C_2(N) = e A N \ln(N) - N M D + R(P+D)    
  \end{displaymath}
  This cost is minimized for:
  \begin{eqnarray*}
    \argmin{N} C_2(N) &=& \argmin{N} e A N \ln(N) = N M D
  \end{eqnarray*}
  The first derivative of which is zero for $\hat{N_2} =
  e^{\frac{MD}{Ae}}$. The second derivative is positive, so we indeed
  have an optimum.
\end{proof}

Our optimizer evaluates both $C_1{\hat{N_1}}$ and $C_2\hat{N_2}$ and
chooses the lower one for the runtime plan.


\section{Experimental Evaluation}\label{sec:experiments} 
In this section, we present our experiments that evaluate the
optimizer described in Section~\ref{sec:optimizer}.  We compare our
approach to Vowpal Wabbit (VW)~\cite{allreduce}: a state of the art
machine learning system.  Our goal here is to verify the theoretical
foundation of our optimizer as it is encoded in
Hyracks.\footnote{Hyracks is available as Open Source Software:
  \url{https://code.google.com/p/hyracks/}} We show that the
time-optimal fan-in is indeed a constant, and independent of the
aggregation time~$A$ or the number of CPUs~$N$.  We present empirical
evidence showing that our static optimizer acurately predicts the
optimal strategy.

\subsection{Task}
Before presenting the results, we first introduce the chosen task:
computing gradients for the training of a large scale linear model.
The goal of training a linear model can be formalized as:
\begin{equation} \label{eq:costfunction} \hat{w} = \argmin{\w} \sum_{\xy \in D}
  l\left(\inner{\vec{x}}{\w},y\right) \end{equation}
where $D$ is the set of tuples of data point~$\vec{x}$ and label~$y$.
The loss function~$l$ measures the empirical loss (divergence) between
the prediction~$\inner{\w}{\vec{x}}$ using the model~$\w$ and the true
label~$y$.  In many cases, it is convex and differentiable in the
prediction, and therefore in the model~$\w$.  Hence, the objective
function~\eqref{eq:costfunction} is amenable to convex optimization.
More precisely, the objective function can be minimized using gradient
descent methods.  Such methods, at their core, perform iterative steps
of the following form:
\begin{equation} \label{eq:update} \w_{t+1} = \w_t - \eta \sum_{(x,y)\in D}
  \delta_w l\left(\inner{\vec{x}}{\w_t},y\right) \end{equation}
Here, $\delta_w$ denotes the gradient with respect to the model~$\w$
and $\eta$ the step size.  The dominant cost in this is computing the
gradients, which decomposes per tuple~$\xy$.  Hence, this task is
amenable to MapReduce and the overall procedure to Iterative
MapReduce.

\textbf{Data Set:} All experiments reported here were performed on a
real-world dataset drawn from the advertisement domain.  The data
consists of 2,319,592,301 records whose feature vectors~$y$ are
sparse, containing a total of 37,113,474,662 non-zero features.  A
textual representation of the data set in the format used by VW (see
below) is 492\,GB in size.

\begin{table}
  \begin{center}
    \begin{tabular}{l l r}
      \small
      Symbol & Meaning & Value \\
      \hline
      $R$      & total \# records & 2,319,592,301 \\
      $N_{max}$ & Max \# map tasks      & 120 \\
      \hline
      $M$ & \# records cached per task & 19,329,936 \\
      \hline
      $P$ & Map time per record        & $3.895 \times 10^{-6}$ s\\
      $D$ & Load time per record       & $w \times 10^{-6}$ s\\
      $A$ & Aggregation time per object    & 2.1 s\\
      \hline
    \end{tabular}
  \end{center}
  \caption{Characteristics of the evaluated environment}
  \label{table:optimizer-input}
\end{table}

\textbf{Cluster: } All experiments were conducted on a single rack of
30 machines in a Yahoo! Research Cluster. Each machine has 2 quad-core
Intel Xeon E5420 processors, 16GB RAM, 1Gbps network interface card,
and four 750GB drives configured as a JBOD, and runs RHEL 5.6. Thus,
each machine can support 4 map tasks, leavings us with $N_{max}=120$.
The machines are connected to a top of rack Cisco 4948E switch.  The
connectivity between any pair of nodes in the cluster is 1Gbps.
Table~\ref{table:optimizer-input} shows the statistics of the dataset
and task which we measured and use as input for our optimizer.

\subsection{Grounding Experiment}
We begin with an experiment that compares our optimized plan, executed
in the Hyracks runtime system, to Vowpal Wabbit (VW)~\cite{allreduce}.
VW uses Hadoop to schedule a map only job.  Each of these map tasks
then downloads the textual data assigned to them from HDFS to the
local disk in an optimized binary format.  The CPUs span a binary
aggregation tree for the reduce operation.  Each CPU emits one result
into the aggregation tree, effectively pre-aggregating the per-CPU
results.  VW is the first system to achieve terra-data scale: It can
operate on datasets with ``trillions of features, billions of
training examples and millions of parameters.''~\cite{allreduce}.  

On the complete dataset using gradients of 128MB ($2^{24}$
dimensions), the average iteration time of VW is 124.41s when run on
all 120~CPUs machines of the cluster.  The average iteration time for
Hyracks in the same configuration (using a binary aggregation tree) is
127.42s.  We performed an additional experiment using a fan-in of 4
and per machine pre-aggregation, which resulted in a average iteration
time of 114.54s.  Hence, our optimized plan beats the current state of
the art for this task.


Our optimizer suggests the use of more CPUs than available to us
(1500) for the given dataset size.  Interestingly, and not by clever
experimental design, it also predicts $N=120$ to be the cost
minimizing configuration for which a cost of $13,700$~CPU seconds is
predicted.  We in fact measure $15,000$, which is remarkably close
given that our optimizer assumes the optimal fan-in of $e$ and not $2$
as used here for comparison with VW.

\subsection{Constant Fan-In}
Our theoretical analysis suggests that the optimal fan-in of an
aggregation tree is independent of both the number of leaf nodes~$N$
and the transfer and processing time per object~$A$.  To evaluate this
claim, we constructed trees with varying fan-in over different numbers
of leaf nodes aggregating different vector sizes.  In
Table~\ref{table:fan-in}, we report the minimum-time fan-in found for
each combination.  The results show the minimum fan-in is constant at
either $4$ or $5$ in the vast majority of cases.  Thus, the
theoretical prediction that the fan-in is a constant, which we have
empirically verified.  However, the empirically found optimum differs
from the theoretical prediction $e$.  We attribute this deviation to
effects not modeled in our theory.  To be precise, the addition of an
aggregation node adds a one-time (setup) cost to the system, which is
amortized via the higher fan-in empirically.

\begin{table}
  \begin{center}
    \begin{tabular}{r c c c c c }
      size/N &2 & 4 & 8 & 16 & 32 \\
      \hline
      1\,MB   & 8 & 5 & 4 & 5  & 4 \\
      2\,MB   & 5 & 3 & 5 & 5  & 5 \\
      4\,MB   & 5 & 5 & 4 & 4  & 4 \\
      8\,MB   & 5 & 4 & 5 & 5  & 3 \\
      16\,MB  & 5 & 4 & 5 & 5  & 5 \\
      32\,MB  & 5 & 5 & 5 & 5  & 3 \\
      64\,MB  & 4 & 4 & 5 & 5  & 5 \\
      128\,MB & 8 & 3 & 5 & 5  & 5 \\
    \end{tabular}
  \end{center}
  \caption{Optimal fan-in for combinations of vector size and number
    of leaf nodes.}
  \label{table:fan-in}
\end{table}

\subsection{Optimal Partitioning}
We now evaluate the other theoretical result presented earlier: a
prescription for the optimal number of machines to use for a given
job.  To create this scenario, we use only $1/5$ of our total dataset,
containing $463,925,403$ records.  This amount of data (roughly 100GB
in text form) can fit in the main memory of a subset of our 120~CPUs.
For the characteristics of our cluster as reported in
Table~\ref{table:optimizer-input}, our optimizer picks $N=N_{max}=120$
to minimize response time and $N=24$ to minimize cost.

Figure~\ref{fig:iteration-time} shows the average iteration times and
costs over this dataset for different numbers of CPUs.  All
experiments use a fan-in of $4$, as determined by the prior
experiment.  The results show that the response time is indeed
minimized for $N=120$, as predicted by our optimizer.  Furthermore,
$N=24$ is the cost minimizing configuration for this job, again as
predicted.

\pgfplotsset{
    width=\columnwidth,
    legend style={
        at={(0.03, 0.97)},
        anchor=north west
    },
}

\begin{figure}
  \centering
   \includegraphics{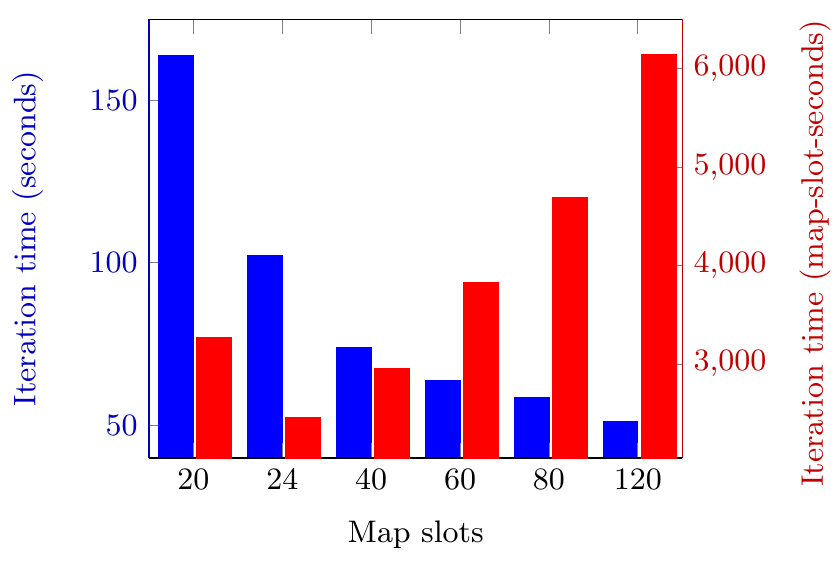}
  \caption{Iteration time and cost using different numbers of CPUs}
  \label{fig:iteration-time}
\end{figure}

\subsection{Discussion}

Our runtime and optimizer is competitive with the current state of the
art in large scale machine learning systems.  This is especially
noteworthy as it makes fewer assumptions than competing systems: It
neither assumes enough resources to cache all data (like Spark), nor
does it default to read all data from disk (like Hadoop).
Additionally, all experimental findings were consistent with the
theoretical findings presented above.  In summary, our static
optimizer was able to pick a good plan in all combinations we tested.


\section{Conclusions}\label{sec:conclusion} MapReduce does not support iteration, which is important for machine
learning tasks that are being increasingly carried out on Big Data in
large-scale ``cloud'' cluster environments.  In this paper, we argued
that the right way to support iteration is to fundamentally extend the
MapReduce model with a looping construct, thereby allowing the system
to reason about the entire program execution.  We presented such an
extension, called Iterative MapReduce.  To illustrate the power of
automatic database-style optimization, we considered a class of
Iterative MapReduce programs that can readily express many ML tasks,
and developed an optimizer that automatically instantiates an
efficient execution plan, taking into account a broad range of
optimizations including data-local and loop-aware scheduling, data
caching, serialization costs, intelligent data partitioning and
resource allocation, and auto-configuration of the aggregation-tree
for the reduce phase.  We presented theoretical justifications for the
two key decisions made by the optimizer on a per-program basis, namely
data partitioning/resource allocation and aggregation-tree
configuration, and presented empirical results that demonstrate our
plans to be competitive with a specialized state-of-the-art
implementation.

Much remains to be done.  The optimizer must be extended to cover the
full range of Iterative MapReduce programs, and to take into account
the likelihood of different kinds of failures in a cost-based manner.
A more comprehensive evaluation must be carried out to establish that
optimizers can indeed be competitive with specialized state-of-the-art
implementations for diverse ML problems.  Nonetheless, our results are
extremely encouraging in that they offer the promise of efficient
system-driven optimization for a broad class of ML programs.  This is
especially significant given that programmers cannot effectively tune
their programs in cloud systems with rapidly changing resource
availability (thanks to multi-tenancy, elasticity, and input datasets
that can change significantly across different runs of the same
program).  We believe that automatic system-driven program
optimization along the lines pioneered by database query optimizers is
the only feasible avenue for future cloud systems, and the results in
this paper are a first step in this direction.

\nocite{Yingyi-Bu:2010fk}
\bibliographystyle{abbrv} \small 
\bibliography{paper}
\end{document}